\documentclass[3p,times, 11pt]{elsarticle}
 \usepackage[latin1]{inputenc}
\usepackage{amssymb,amsmath,amsthm,amscd,latexsym}
\usepackage{mathrsfs}
\usepackage{mathrsfs}
\usepackage{mathrsfs}
\usepackage{amsfonts}
\usepackage{amsmath}
\usepackage{amssymb}
\usepackage{euscript}
\usepackage{amscd}
\usepackage{multirow}
\usepackage{tipa}
\usepackage{array}
\usepackage{graphicx}
\usepackage{hyperref}
\hypersetup{
    colorlinks=true,
    linkcolor=blue,
    citecolor=blue,
    urlcolor=blue
}

\DeclareMathAlphabet{\mathpzc}{OT1}{pzc}{m}{it}
\usepackage{xypic}
\usepackage{epstopdf}
\usepackage{yfonts}
\usepackage[T1]{fontenc}
\usepackage[all]{xy}
\usepackage{times}
\usepackage{color}
\usepackage{caption}
\usepackage{listings}
\usepackage{listings}
\usepackage{mathtools}
\usepackage{tikz}
\usepackage{mathtools}
\usepackage{pstricks}
\usepackage{tikz}
\usetikzlibrary{fit}
\usepackage{mathdots}
\usepackage{amsmath,subcaption}
\usepackage{tikz-cd}
\usepackage{mathtools, amssymb}
 \usepackage{pst-node, auto-pst-pdf}
\usetikzlibrary{matrix}
\lstdefinelanguage{Magma}
{
  morekeywords={function,end,function,return,while,for,if,else,then},
  sensitive=true,
  morecomment=[l]{//},
  morecomment=[s]{/*}{*/},
  morestring=[b]",
}
\usepackage[linesnumbered,ruled,vlined]{algorithm2e}

\newcommand{\forceindent}{\leavevmode{\parindent=1em\indent}}

\usepackage{mathrsfs}

\newtheorem{Definition}{Definition}
\newtheorem{Theorem}{Theorem}

\newtheorem{Remark}{Remark}

\newtheorem{Definition and Notation}{Definition and Notation}
\newtheorem{Lemma}{Lemma}

\newtheorem{Proposition}{Proposition}
\newtheorem{Corollary}{Corollary}

\usepackage{amssymb}

\usepackage[figuresright]{rotating}

\begin{document}

\begin{frontmatter}

\title{$(\sigma,\delta)$-polycyclic codes in  Ore extensions over rings}

\author[1,m]{Maryam Bajalan }  
\ead{maryam.bajalan@math.bas.bg}
\fntext[m]{Maryam Bajalan is supported by the Bulgarian Ministry of Education and Science, Scientific Programme "Enhancing the Research Capacity in Mathematical Sciences (PIKOM)", No. DO1-67/05.05.2022.
}
\author[2]{Ivan Landjev}
\ead{i.landjev@nbu.bg}

\author[3,v]{Edgar Mart\'inez-Moro} 
\ead{edgar.martinez@uva.es}
\fntext[v]{Edgar Mart\'inez-Moro is partially supported by Grant TED2021-130358B-I00 funded by MCIN/AEI/10.13039/501100011033 and by the ``European Union NextGenerationEU/PRTR''
}

\author[4]{Steve Szabo} 
\ead{steve.szabo@eku.edu}

\address[1]{Institute of Mathematics and Informatics, Bulgarian Academy of Sciences, Bl. 8, Acad. G. Bonchev Str., 1113, Sofia, Bulgaria}
\address[2]{Department of Mathematics, New Bulgarian University, Sofia, Bulgaria}
\address[3]{Institute of Mathematics, University of Valladolid, Castilla, Spain}
\address[4]{Institute of Mathematics and Statistics, Eastern Kentucky University,  Kentucky, USA}

\begin{abstract}
In this paper, we study the algebraic structure of $(\sigma,\delta)$-polycyclic codes, defined as submodules in the quotient module $S/Sf$, where $S=R[x,\sigma,\delta]$ is the Ore extension ring, $f\in S$, and  $R$ is a finite but not necessarily commutative ring.  We establish that the Euclidean duals of $(\sigma,\delta)$-polycyclic codes are $(\sigma,\delta)$-sequential codes.  By using $(\sigma,\delta)$-Pseudo Linear Transformation, we define the annihilator dual of $(\sigma,\delta)$-polycyclic codes. Then, we demonstrate that the annihilator duals of $(\sigma,\delta)$-polycyclic codes maintain their $(\sigma,\delta)$-polycyclic nature. Furthermore, we classify when two $(\sigma,\delta)$-polycyclic codes are Hamming isometrical equivalent. By employing Wedderburn polynomials, we introduce simple-root $(\sigma,\delta)$-polycyclic codes. Subsequently, we define the $(\sigma, \delta)$-Mattson-Solomon transform for this class of codes and we address the problem of decomposing these codes by using the properties of Wedderburn polynomials.

\end{abstract}
 

\end{frontmatter}

\section{Introduction}

\forceindent Suppose that $R$ is a finite ring with unit,  $\sigma\in \operatorname{End}(R)$   and $\delta$  is an  additive endomorphism of $R$ satisfying  $\delta(ab) = \sigma(a)\delta(b) + \delta(a)b$ for all $a,b \in R.$    An Ore extension is  a non-commutative polynomial ring  $R[x,\sigma,\delta]$  constructed by    multiplications    $xa = \sigma(a)x+\delta(a)$ and $ x^rx^s=x^{r+s}$.
In coding theory, $(\sigma,\delta)$-polycyclic codes over the ring $R$ is a novel study, where these codes are defined as ideals (submodules) in the quotient ring (quotient module) $S/Sf$ where $S$ is the Ore extension   $R[x,\sigma,\delta]$  and $Sf$ as a two-sided (one-sided) ideal. This class of codes constitutes a powerful generalization that encompasses various generalizations of cyclic codes,  including skew cyclic codes,  skew negacyclic codes, skew constacyclic codes, polycyclic codes, etc, see \cite{our1, our, Ulmer, B-U-3, B-U-4, siap}.  Therefore, by studying $(\sigma,\delta)$-polycyclic codes from several points of view,  we achieve a unification of different approaches, providing a comprehensive survey of all these various generalizations. 

\forceindent In $2012$, Boucher and Ulmer have introduced the concept of $(\sigma,\delta)$-polycyclic codes as submodules in $S/Sf$, where $R$ is the finite field $\mathbb F_q$, see \cite{B-U-5}.
 It is well-known that if $R=\mathbb{F}_q$, then all $\sigma$-derivations $\delta$ are uniquely determined by $\delta(a) = \beta(\sigma(a)-a)$, where $\beta \in \mathbb{F}_q$. They have proved that if $z=x+\beta$, then there exists a bijective ring homomorphism between the Ore extension $\mathbb F_q[x, \sigma,\delta]$ and the skew polynomial ring $\mathbb F_q [z,\sigma]$.   However, this homomorphism preserves neither the Hamming distance nor the rank distance. As a result, $(\sigma,\delta)$-polycyclic codes in  $\mathbb{F}_q[x, \sigma, \delta]$ produce new codes that are not isometrically equivalent to skew cyclic codes in the skew polynomial ring $\mathbb F_q [z,\sigma]$.  Additionally, they have demonstrated that $(\sigma,\delta)$-polycyclic codes can provide improved distance bounds compared to skew cyclic codes in  $\mathbb F_q[x,\sigma]$.
 
\forceindent There is a limited account on $(\sigma, \delta)$-polycyclic codes. In \cite{Boulagouaz}, the authors determined the generator and parity-check matrices of $(\sigma,\delta)$-polycyclic codes over rings. In \cite{Tironi}, it is proven that  $(\sigma,\delta)$-polycyclic codes over $\mathbb F_q$  are $T_{C_f}$-invariant linear codes in $\mathbb F_q^n$, where $T_{C_f}$ represents  a $(\sigma,\delta)$-Pseudo Linear Transfromation, and $C_f$ is the companion matrix associated with $f$. Under a special assumption for $f$, they considered $(\sigma,\delta)$-polycyclic codes as invariant subspaces and generalized some of the main results in \cite{our, Radkova}. In \cite{Ou}, the authors took a step further and introduced a generalization of the $(\sigma,\delta)$-Pseudo Linear Transformation $T_{C_f}$ into the transformation $T_M$, with $M$ being an arbitrary matrix. Then, they generalized the $(\sigma,\delta)$-polycyclic codes over fields into $T_M$-invariant linear codes. As applications, in \cite{Parish1, Parish2}, $(\sigma, \delta)$-polycyclic codes were used to establish good quantum codes, and in \cite{algebra}, these codes were applied to construct $\mathbb Z$-lattices in $\mathbb R^N$. It is remarkable that $(\sigma, \delta)$-polycyclic codes have been described by different terminologies in the above references, such as ``$(\sigma, \delta)$-codes'',  ``$(f, \sigma, \delta)$-codes'',  ``$(M, \sigma, \delta)$-codes'',  and ``skew generalized cyclic codes''. All these terms represent one concept: the generalization of polycyclic codes (see \cite{our}) into the non-commutative polynomial ring $R[x,\sigma,\delta]$. In this paper, we use the term ``$( \sigma, \delta)$-polycyclic codes'' for this concept.
 
\forceindent Motivated by these seminal works,  this paper aims to investigate   $(\sigma,\delta)$-polycyclic codes as submodules in $S/Sf$, where $R$ is a finite, but not necessarily commutative, ring. This is the first attempt to study this class of codes in the non-commutative ring $R$.
A challenging aspect of $(\sigma,\delta)$-polycyclic codes is that their Euclidean duals do not have $(\sigma,\delta)$-polycyclic code structure. To address this, we employ $(\sigma,\delta)$-Pseudo Linear Transformations to introduce the non-degenerate sesquilinear form $\langle\,,\,\rangle_0$.  We then establish that 
the left and right duals related to this form preserve their $(\sigma,\delta)$-polycyclic code nature. 

\forceindent In coding theory, a well-known problem is the classification of isometrically equivalent codes. We establish that if there exists a monomial matrix $B$ such that $C_{f_1}B = \sigma(B)C_{f_2} + \delta(B)$, then $(\sigma,\delta)$-polycyclic codes in $S/Sf_1$ and $(\sigma,\delta)$-polycyclic codes in $S/Sf_2$ are Hamming isometrically equivalent.

\forceindent The Mattson-Solomon transform is a  well-known mathematical technique with diverse applications, including the determination of minimal generator polynomials and Reed-Solomon codes,  see \cite{our, san, Sloan}.
Our next challenge is defining the Mattson-Solomon transform in the Ore extension setting. To achieve this, we employ the theory of Wedderburn polynomials, which is extensively explored by Lam and Leroy in \cite{lamleroy1, lamleroy2} when $R$ is a  division ring. We classify a $(\sigma,\delta)$-polycyclic code as simple-root if $f$ is a  Wedderburn polynomial. Note that in the case where $R$ is a field,  the Wedderburn polynomials have the form $(x-a_1)\cdots(x-a_n)$, where $a_i$s are distinct elements, see \cite{lamleroy1}. 
Therefore, the simple-root $(\sigma,\delta)$-polycyclic codes are the generalization of simple-root polycyclic codes, as defined in \cite{our}. Our motivation in defining the $(\sigma, \delta)$-Mattson-Solomon transform for these specific codes lies in the invertibility of the $(\sigma,\delta)$-Vandermond matrix, see \cite{lamleroy2} Theorem 5.8.

\forceindent  The next challenge is the decomposition of  $(\sigma,\delta)$-polycyclic codes. Utilizing the Wedderburn polynomial,  we define the simple-root $(\sigma,\delta)$-polycyclic code $C$ and decompose them into the direct sum of the simple-root $(\sigma,\delta)$-polycyclic codes $C_i$ such that each  $C_i$  is related to an eigenvalue of $T_{C_f}$.

\forceindent The structure of this paper can be outlined as follows.
In Section \ref{590}, the foundation of the Ore extension rings for the subsequent sections is introduced.  In Sections \ref{S1} and \ref{S2},  the concept of right and left $(\sigma, \delta)$-polycyclic and sequential codes and their relationships are presented.  In Section \ref{S3},  an alternate concept of left and right dualities is offered, ensuring the dual of any $(\sigma,\delta)$-polycyclic code remains a $(\sigma,\delta)$-polycyclic code.  In Section \ref{S4}, the  Hamming isometry equivalence of $(\sigma,\delta)$-polycyclic codes is explored. In Section \ref{26}, the $(\sigma,\delta)$-Mattson-Solomon maps for simple-root $(\sigma,\delta)$-polycyclic codes are presented. Finally, in Section \ref{MM5}, a decomposition of a simple-root $(\sigma,\delta)$-polycyclic code is provided.
\section{Preliminary}\label{590}
\subsection{Ore extension}

Suppose  $R$ is a finite but not necessarily commutative ring with unit  $1\neq 0$,  $\sigma\in \operatorname{End}(R)$   and $\delta$ a left $\sigma$-derivation, meaning  that $\delta$ is an  additive endomorphism of $R$ satisfying  $\delta(ab) = \sigma(a)\delta(b) + \delta(a)b$ for all $a,b \in R.$   The set of all polynomials with coefficients
on the left  forms a non-commutative ring  under the usual polynomial addition  and the multiplications given by  
\begin{equation}\label{18}
    xa = \sigma(a)x+\delta(a)\quad \text{for all}\, a \in R\qquad \text{and}\qquad x^rx^s=x^{r+s}  \quad \text{for all}\, r,s\in \mathbb{N}.
\end{equation}
It is well-known that for every $k\in \mathbb N$
\begin{equation}\label{9}
x^k a = \sum_{i=0}^k \binom{k}{i} \sigma^i\delta^{k-i}(a) x^{i} \quad \text{for all}\, a \in R.
\end{equation}
This ring is known as the \textit{Ore extension}  and we denote it by $S = R[x, \sigma, \delta]$. Note that  $\sigma(1)=1$ and $\delta(1)=0$, which imply that $1$ is an identity for $S$ as well. If $(\sigma,\delta)=(\operatorname{Id},0)$, then the Ore extension is the usual polynomial ring $R[x]$. 
In the case where $\delta=0$, the Ore extension is called  \textit{skew polynomial ring}  and is denoted by $R[x,\sigma]$.  Moreover, when $\sigma=0,$ it is called \textit{differential polynomial ring} and is denoted by $R[x,\delta].$ In coding theory literature, the term  ``skew polynomial ring''  is often used for both $R[x,\sigma]$  and $R[x, \sigma, \delta]$. However, to avoid confusion, we distinguish between the two cases and refer to the term ``skew polynomial ring'' only when $\delta=0$,  while for the broader context of $R[x, \sigma, \delta]$, we employ the term ``Ore extension''. A comprehensive overview of the Ore extensions over rings can be found in  \cite{ Gomez, lamleroyhom, Hilbert90, lamleroy3,   lamleroy1, lamleroy2,  Mc,  ore,  Richter}. This section aims to highlight fundamental characteristics for use in the next sections.

If ${\delta}(ab) = {\delta}(a){\sigma}(b) + a{\delta}(b)$, it is called  right ${\sigma}$-derivation. Note that a right  ${\sigma}$-derivation over $R$ is a left  ${\sigma}$-derivation over the opposite ring $R^{op}$ (the ring constructed by $R$ with the same underlying Abelian group but with the multiplication operation reversed). Throughout this paper, our focus is on left ${\sigma}$-{derivations}, and we simply refer to them as  ${\sigma}$-{derivations}, unless stated otherwise. 
According to our definition of Ore extension, $S=R[x,\sigma,\delta]$ is a left $R$-module with the basis $\mathcal B=\{x^n: n \geqslant 0\}$. Now, if we want to consider $S$ as the set of all polynomials with coefficients
on the right and under the usual polynomial addition and the multiplications given by  $x^rx^s=x^{r+s}$ for $r,s\in \mathbb{N}$ and $ax =x{\sigma}(a)+{\delta}(a)$ {for} $ a \in R,$ then the ring $S$ is referred to as the \textit{opposite Ore extension}. If $\sigma\in\operatorname{Aut}(R)$, then $-\delta\sigma^{-1}$ is a left $\sigma^{-1}$-derivation of $R^{op}$ and $S^{op}=(R[x,\sigma,\delta])^{op}\cong R^{op}[x,\sigma^{-1},-\delta\sigma^{-1}]$, see \cite{Torrecillas, evaluation}.

The ring $S$ has a division algorithm that works for left and right divisors. More precisely, given any two non-zero polynomials $ g, k\in S$, there exist polynomials  $q,r \in S$ such that $g = qk + r$ and $\deg(r) < \deg(k)$, see \cite{Gomez} Theorem 4.34. Consider a sequence of mappings $N_i : R \to R$ such that $N_0(a) = 1$ and $N_{i+1}(a) =  \sigma(N_i(a))a+\delta(N_i(a))$ for  $a \in R$. If $g(x)=\sum_{i=0}^{n} g_i x^i\in S$ and  $g(x)= q(x - a) + r$ then, $r =\sum_{i=0}^n g_iN_i(a)=g(a)$, see \cite{vandermond} Lemma 2.4. Therefore, $g(a)=0$ if and only if  $x -a$ is a right factor of $g(x)$. In the case $g(a) = 0$,  $a$ is called the right root of $g$. Furthermore, when $\sigma$ is injective and $R$ is a division ring, then $S$ is a left PID,  and thereby it is a UFD, see \cite{Gomez} Corollary 4.36 and Theorem 4.25.

Two elements $a$ and $b$ in $R$ are called $(\sigma,\delta)$-\textit{conjugate} if there exists an invertible element $c$ in $R^{\times}$ such that $b = \sigma(c)ac^{-1}+\delta(c)c^{-1}$. 
It is straightforward to verify that $(a^c)^d=a^{dc}$, which implies that the relation is an equivalence relation on $R$, see \cite{vandermond}.
We will denote $\sigma(c)ac^{-1}+\delta(c)c^{-1}$ by $a^c$ and  the $(\sigma,\delta)$-conjugacy class of $a$ by $\Delta(a) = \{a^c :  c \in R^{\times}\}.$

Throughout the paper, we assume that $f\in S$ is the monic polynomial $f(x) =-f_0-f_1x -f_{n-1}x^{n-1}+x^n$ and we  use $Sf$ to represent the left $S$-submodule generated by $f$. If we consider the set $\{g\in S: \deg{g}<n\}$ with the usual addition and the multiplication $gh\, ( \operatorname{mod} f)$, then this set has the structure of   a left  free $R$-module  with the $R$-basis $\mathcal B=\{1, x, \ldots, x^{n-1}\}$. We denote this module by $S/Sf$. If $Sf$ is a two-sided ideal in $S$, then $S/Sf$ also possesses a ring structure.

\subsection{$(\sigma, \delta)$-Pseudo Linear Transformation}\label{60001}
 Let  $V$ be a left $R$-module.  An additive  map $T : V \rightarrow V$  is called   $(\sigma, \delta)$-\textit{Pseudo Linear Transformation} (or $(\sigma, \delta)$-PLT for short) if 
$$T (av) = \sigma(a) T (v) + \delta(a) v\quad \text{ for all } a \in R.$$
For more details on  $(\sigma, \delta)$-PLT, the reader is referred to \cite{Bronstein,  evaluation, LEROY}.
Suppose  $V$ is a free left $R$-module with  a  basis $\{e_1, e_2, \ldots, e_n\}$.
For any  matrix $C\in M_{n}(R)$,    the map $T_C: V\to V$ defined as  $T_{C}(v)=  \sigma(v)C+\delta(v),$ where $\sigma(v)$ and $\delta(v)$ are the extensions of $\sigma$ and $\delta$ from $R$ to $R^n$, is  a $(\sigma, \delta)$-PLT.  In particular, for $n = 1$ and $a \in R$, the map
$T_a : R \to R$ given by $T_a(x) = \sigma(x)a + \delta(x)$ is a $(\sigma, \delta)$-PLT.  Clearly, $T_0 = \delta$.


\begin{Proposition}[\cite{evaluation} Lemma 2 and \cite{LEROY} Proposition 1.2]\label{5}
The following conditions are equivalent for an additive group $(V, +)$.
\begin{enumerate}
\item $V$ is a left $R$-module and there exists a $(\sigma,\delta)$-Pseudo Linear Transformation $T:V\rightarrow V.$
\item $V$ is a left $S=R[x, \sigma, \delta ]$-module induced by the product $g(x)\cdot v := g(T )(v).$
\item There exists a ring homomorphism $\Lambda: S \to End(V, +)$ such that  for $g(x) = \sum_{i=0}^{n-1} g_i x^i \in S$, $\Lambda (g):=g(T)=\sum_{i=0}^{n-1} g_i T^i \in \operatorname{End}(V, +).$ 
\end{enumerate}
\end{Proposition}
\begin{Corollary}[\cite{LEROY} Corollary 1.3]\label{86}
For any $g, h \in  S$ and any $(\sigma,\delta)$-Pseudo Linear Transformation $T$, $(gh)(T) = g(T) h(T)$.
\end{Corollary}

In the classical polynomial ring, evaluating a polynomial at a point is obtained by substituting the variable x with the point. However, the process is different in the Ore extension.

\begin{Lemma}[\cite{evaluation} Corollary 2.10]\label{20}
For  $g, h \in S$ and $a \in R$,  $gh(a) = g(T_a)(h(a)).$ In particularly, if $h(a)$ is invertible then $(gh)(a) =g(a^{h(a)})h(a)$, and if $h(a)=0$ then  $(gh)(a)=0.$
\end{Lemma}
\section{$(\sigma,\delta)$-polycyclic codes}\label{S1}
 A code $C$ of length $n$ over the ring  $R$ is a subset of $R^n$. A left (right) linear code $C$ of length $n$ over $R$ is a  left (right) $R$-submodule of $R^n$. Let $f(x)=-f_0-f_1x-\cdots-f_{n-1}x^{n-1}+x^n\in S$ be a monic polynomial.  
\begin{Definition}\label{11301} Let $C$ be a left linear code over $R$. 
\begin{enumerate}
    \item The code $C$ is called  right $(\sigma,\delta)$-polycyclic  if for every $g=(g_0,g_1,\dots,g_{n-1}) \in C$ we have 
 \begin{equation}\label{301}
     \bigg(0, \,\,\sigma(g_0),\,\, \sigma(g_1), \,\,\ldots,\,\, \sigma(g_{n-2})\bigg)+\sigma(g_{n-1})\bigg(f_0,f_1,\ldots, f_{n-1}\bigg)+\delta(g)\in C.
 \end{equation}
 \item The code $C$ is called left $(\sigma,\delta)$-polycyclic if for every $g=(g_0,g_1,\dots,g_{n-1}) \in C$ we have 
    \begin{equation}\label{1301}
        \bigg(\sigma(g_1), \,\, \sigma(g_2),\,\, \dots, \,\, \sigma(g_{n-1}), \,\,0\bigg) + \sigma(g_0)\bigg(f_0,f_1,\dots,f_{n-1}\bigg)+\delta(g) \in C.
    \end{equation}
\end{enumerate}
   \end{Definition}
\begin{Remark}
    In the case $(\sigma,\delta)=(\operatorname{Id}, 0),$  right and left $(\sigma,\delta)$-polycyclic codes coincide with right and left polycyclic codes, respectively, see for example  \cite{Szabo}.
\end{Remark} 
Let us denote the companion matrix related to $f$ as
$$C_f=\begin{pmatrix}
   0    &   1   &   0     & \cdots &   0   \\
   0    &   0   &   1      & \cdots &   0   \\
   \vdots & \vdots & \ddots  & \cdots & \vdots \\
   0    &   0   &   0     & \ddots &   1   \\
  f_0 & f_1 & f_2 & \cdots & f_{n-1}
\end{pmatrix}.
$$

\begin{Lemma}\label{102}
 The  code $C$ is a right $(\sigma, \delta)$-polycyclic code if and only if it is invariant under $T_{C_f}=\sigma(v)C_f+\delta(v)$. 
\end{Lemma}
\begin{proof}
    For each $g=(g_0,g_1,\ldots, g_{n-1})\in R^n,$ it is easy to see that 
    \begin{align*}
       T_{C_f}(g)& =\sigma(g)C_f+\delta(g)\\
       &= \bigg( \sigma(g_{n-1})f_0, \sigma(g_0)+ \sigma(g_{n-1})f_1, \sigma( g_1 )+ \sigma(g_{n-1})f_2, \ldots, \sigma(g_{n-2})+ \sigma(g_{n-1})f_n\bigg)+\delta(g)\\
       &=\bigg(0, \sigma(g_0), \sigma(g_1), \ldots, \sigma(g_{n-2})\bigg)+\sigma(g_{n-1})\bigg(f_0,f_1,\ldots, f_{n-1}\bigg)+\delta(g).
       \end{align*} 
    Therefore, $C$  is a right  $(\sigma, \delta)$-polycyclic code if and only if $T_{C_f}(C)\subseteq C.$
\end{proof}
The left $R$-module isomorphism  $\Omega_f : R^n \rightarrow S/Sf$ associates  each element $g=(g_0,g_1,\ldots,g_{n-1})$ to  its corresponding polynomial  $g(x)=\sum_{i=0}^{n-1}g_ix^i.$ 
 \begin{Lemma}\label{75}
     The $(\sigma,\delta)$-Pseudo Linear Transformation $T_{C_f}$ corresponds to the left multiplication by $x$ on $S/Sf.$
 \end{Lemma}
 \begin{proof}
     Let  $g(x)$ be the image of $g = (g_0, g_1, \ldots, g_{n-1})\in R^n$ under $\Omega_f.$ Then  
     \begin{align*}
    xg(x) &= \sum_{i=0}^{n-1} x g_i x^i =\sum_{i=0}^{n-1} \bigg(\sigma(g_i)x + \delta(g_i)\bigg) x^i \\
    &= \sigma(g_0)x + \delta(g_0) + \sigma(g_1)x^2 + \delta(g_1)x + \sigma(g_2)x^3 + \delta(g_2)x^2 + \cdots + \sigma(g_{n-1})x^n + \delta(g_{n-1})x^{n-1} \\
    &= \bigg(\sigma(g_{n-1})f_0 \bigg) + \bigg(\sigma(g_{n-1})f_1+\sigma(g_0) \bigg)x + \cdots + \bigg(\sigma(g_{n-1})f_{n-1}+\sigma(g_{n-2}) \bigg)x^{n-1}+ \sum_{i=0}^{n-1}\delta(g_i)x^i.
        \end{align*}
It is easy to see  that  $ \bigg(\sigma(g_{n-1})f_0,\,\, \sigma(g_{n-1})f_1+\sigma(g_0),\,\, \ldots,\,\, \sigma(g_{n-1})f_{n-1}+\sigma(g_{n-2}) \bigg)=\sigma(g)C_f$.  Consequently,   $xg(x)$ corresponds to the image of $\sigma(g)C_f+\delta(g)=T_{C_f}(g)$  under the mapping $\Omega_f,$ which completes the proof.
 \end{proof}
\begin{Lemma}\label{80}
  The  code $C$ is a right  $(\sigma, \delta)$-polycyclic code  if and only if $\Omega_f(C)$ is a left $S$-submodule of $S/Sf$.
\end{Lemma}
\begin{proof}
    Assume that   $C$ is  a $(\sigma, \delta)$-polycyclic code.   Choose an arbitrary element  $g(x)=\sum_{i=0}^{n-1}g_ix^i \in \Omega_f(C)$ and denote its corresponding element by  $g = (g_0, g_1, \ldots, g_{n-1})\in C$. Applying   Lemma \ref{75} and the fact that $T_{C_f}(g)\in C$, we have  $xg(x)\in \Omega_f(C).$ By induction,   $x^i g(x)$ also belongs to $\Omega_f(C)$ for all $i$, which implies   $h(x)g(x)\in  \Omega_f(C)$ for any $h(x) \in S$. Hence,  $\Omega_f(C)$ is  a left $S$-submodule of $S/Sf$. Using the same argument, the converse can be proved.
        \end{proof}
From now on, let us denote $\Omega_f(C)$ simply as $C$. Moreover, elements of any $(\sigma, \delta)$-polycyclic code will be denoted interchangeably by both $g$ and $g(x)$.

The matrix  $A \in M_n(R)$ is called  $(\sigma,\delta)$-\textit{similar} to the matrix $B \in M_n(R)$ if there exists an invertible matrix $P$ such that $B = \sigma(P)AP^{-1} + \delta(P)P^{-1}$, where $ \sigma(P)$ and $ \delta(P)$ are the extensions of $\sigma$ and $\delta$ from the ring $R$ to $ M_n(R)$ by applying them on each entry. Furthermore,  the  matrix $A$ is called  $(\sigma,\delta)$-\textit{diagonalizable} if it is $(\sigma,\delta)$-similar  to a diagonal  matrix, see \cite{evaluation}.

\begin{Lemma}\label{350}
Let $M$ be a $(\sigma,\delta)$-similar matrix to  $C_f$  and $C\subseteq R^n.$ Then  
     $C$ is invariant under  $T_M$ if and only if $CP^{-1}$ is invariant under $T_{C_f},$ or equivalently,  $CP$ is invariant under $T_M$ if and only if $C$ is   invariant under  $T_{C_f}$.
\end{Lemma}

\begin{proof}  
Let $C$ be invariant under  $T_M$. Then
 \begin{align*}
T_{C_{f}}(CP^{-1})&= \sigma(CP^{-1})C_{f}+\delta(CP^{-1}) \\
&= \sigma(CP^{-1})\bigg( \sigma(P)MP^{-1} + \delta(P)P^{-1} \bigg)+\delta(CP^{-1})\\
&= \sigma(CP^{-1})\sigma(P)MP^{-1}+\sigma(CP^{-1})\delta(P)P^{-1}+\delta(CP^{-1})\\
&= \sigma(C)MP^{-1}+\delta(CP^{-1}P)P^{-1}\\
&= \sigma(C)MP^{-1}+\delta(C)P^{-1}\\
&= \bigg(\sigma(C)M+\delta(C)\bigg)P^{-1}\subseteq CP^{-1}.\\
\end{align*}
Moreover,
\begin{align*}
    T_{M}(C)&=T_M(CP^{-1}P)=\sigma(CP^{-1}P)M+\delta(CP^{-1}P)\\
    &=\sigma(CP^{-1})\sigma(P)M+\sigma(CP^{-1})\delta(P)+\delta(CP^{-1})P\\
   &=\sigma(CP^{-1})\bigg(\sigma(P)M+\delta(P)\bigg)+\delta(CP^{-1})P\\
&=\sigma(CP^{-1})\bigg(\sigma(P)MP^{-1}+\delta(P)P^{-1}\bigg)P+\delta(CP^{-1})P\\
    &=\bigg(\sigma(CP^{-1})(C_f)+\delta(CP^{-1})\bigg)P\\
    &\subseteq (CP^{-1})P=C.
\end{align*}
\end{proof}
Associated with $f(x) =-f_0-f_1x -f_{n-1}x^{n-1}+x^n$, define the matrix
 \[
E_f = \begin{pmatrix}
  f_0 & f_1 & f_2 & \cdots & f_{n-1} \\
  1 & 0 & 0 & \cdots & 0 \\
  0 & 1 & 0 & \cdots & 0 \\
  \vdots & \vdots & \ddots & \ddots & \vdots \\
  0 & 0 & \cdots & 1 & 0 
\end{pmatrix}.
\]
Similar to the  Lemma \ref{102}, we have the following lemma.
\begin{Lemma}
     A left linear code $C$ is a left $(\sigma, \delta)$-polycyclic code if and only if it is invariant under $T_{E_f}$.
\end{Lemma}
\begin{Lemma}\label{man11}
 The matrix $E_d$ is $(\sigma,\delta)$-similar to $C_f$, where $d=-f_{n-1}-f_{n-2}x \cdots-f_0x^{n-1} +x^n$.
\end{Lemma}
\begin{proof}
It is easy to see that 
   $C_f=PE_dP,$  where 
\begin{equation}\label{man2}
  P = \begin{pmatrix}
0  & 0 & \cdots & 1 \\
0 & 0 & 1 & 0 \\
\vdots & \iddots & \vdots & \vdots\\
1 &  \cdots & 0 & 0 \\
\end{pmatrix}.  
\end{equation}
Note that $P$ is invertible,   $P^{-1}=P$, $\sigma(P)=P$, and $\delta(P)=0,$ which imply 
$C_fP=\sigma(P)E_d+\delta(P).$ 
\end{proof}
\begin{Corollary}\label{man12}
 If $C$ is a right $(\sigma,\delta)$-polycyclic code which is invariant under $T_{C_f}$, then $CP$ is a left $(\sigma,\delta)$-polycyclic code which is invariant under $T_{E_d}$.     
\end{Corollary}

We define the left $R$-module isomorphism  ${\Omega}^{*}_{ d} : R^n \rightarrow {S}/{S}{d}$ which corresponds   each element $ g=(g_0,g_1,\ldots,g_{n-1})\in R^n$ to  its reciprocal polynomial    $g^*(x)=g_{n-1}+g_{n-2}x+\cdots+g_1x^{n-2}+g_0x^{n-1}.$

\begin{Lemma}\label{man10}
The  $(\sigma,\delta)$-Pseudo Linear Transformation $T_{E_d}(g)=\sigma(g)E_d+\delta(g)$ corresponds to the left multiplication $g^*(x)$ by $x$ in   $S/Sd$, but under the reciprocal correspondence.
 \end{Lemma}
 
 \begin{proof}
 Note that in the factor ring  $S/S{d}$, we have $x^n=f_{n-1}+f_{n-2}x+\cdots+f_1 x^{n-2}+f_0 x^{n-1}.$ We get 
     \begin{align*}
    xg^{*}(x) &= \sum_{i=0}^{n-1}   xg_{n-1-i}x^i  =\sum_{i=0}^{n-1} \bigg(\sigma(g_{n-1-i})x + \delta(g_{n-1-i})\bigg)x^i  \\
    &= \sigma(g_{n-1})x + \delta(g_{n-1}) + \sigma(g_{n-2})x^2 + \delta(g_{n-2})x + \sigma(g_{n-3})x^3 + \delta(g_{n-3})x^2+ \cdots+ \sigma(g_0)x^n + \delta(g_0)x^{n-1}\\
    &= \bigg(f_{n-1}\sigma(g_0) \bigg) + \bigg(f_{n-2}\sigma(g_0)+\sigma(g_{n-1}) \bigg)x + \cdots + \bigg(f_{1}\sigma(g_{0})+\sigma(g_{2}) \bigg)x^{n-2}+\bigg(f_{0}\sigma(g_{0})+\sigma(g_{1}) \bigg)x^{n-1}\\
&\qquad\qquad\qquad\qquad\qquad\qquad\qquad\qquad\qquad\qquad\qquad\qquad\qquad\qquad\qquad\qquad  + \sum_{i=0}^{n-1}\delta(g_{n-1-i})x^i.
        \end{align*}
It is easy to see  that $\bigg( f_{0}\sigma(g_{0})+\sigma(g_{1}),\,f_{1}\sigma(g_{0})+\sigma(g_{2}),  \,\ldots,\, f_{n-2}\sigma(g_0)+\sigma(g_{n-1}),\, f_{n-1}\sigma(g_0)\bigg)={\sigma(g)}E_d.$
Therefore, the reciprocal  of $xg^{*}(x)$ is associated with the image of $\sigma(g)E_d+\delta(g)=T_{E_d}(g)$. 
 \end{proof}

\begin{Corollary}
    A code $C$ is a left $(\sigma,\delta)$-polycyclic code  if and only if $\Omega^*_d(C)$ is a left $S$-submodule of $S/Sd$. Moreover, if $\sigma\in\operatorname{Aut}(R)$, then $C$ is a left $(\sigma,\delta)$-polycyclic code  if and only if  it is  a right 
$R^{op}[x,\sigma^{-1},-\delta\sigma^{-1}]$-submodule of $R^{op}[x,\sigma^{-1},-\delta\sigma^{-1}]/R^{op}[x,\sigma^{-1},-\delta\sigma^{-1}]d^{op},$ where   $d^{op}$ is obtained by considering  $d$ with coefficients on the right.
  \end{Corollary}
     \begin{proof}
     The first statement is an immediate consequence of Lemma \ref{man10}. The second statement  follows from the facts that a left $S$-submodule of $S/Sd$ is equivalent to  a right $S^{op}$-submodules $S^{op}/S^{op}d$,  and the isomorphism
         $(R[x,\sigma,\delta])^{op}\cong R^{op}[x,\sigma^{-1},-\delta\sigma^{-1}].$
     \end{proof}




\section{$(\sigma,\delta)$-sequential codes}\label{S2}
\begin{Definition} Let $C$ be a left linear code of length $n$ over  $R$.

\begin{enumerate}
    \item The code $C$ is called right $(\sigma,\delta)$-sequential if for every $(g_0,g_1,\dots,g_{n-1}) \in C$ we have  
\begin{equation}\label{303}
    \bigg(\sigma(g_1),\,\, \sigma(g_2), \,\, \dots, \,\, \sigma(g_{n-2}), \,\,\sigma(g_0)f_0 + \sigma(g_1)f_1 + \cdots + \sigma(g_{n-1})f_{n-1}\bigg)+\delta(g) \in C.
\end{equation}
\item The code $C$ is called left  $(\sigma,\delta)$-sequential if for every $(g_0,g_1,\dots,g_{n-1}) \in C$ we have 
\begin{equation}\label{304}
    \bigg(\sigma(g_0)f_0 + \sigma(g_1)f_1 + \cdots + \sigma(g_{n-1})f_{n-1},\,\, \sigma(g_0),\,\, \sigma(g_2), \,\,\dots,\,\, \sigma(g_{n-2})\bigg)+\delta(g) \in C.
\end{equation}
\end{enumerate}  
  \end{Definition}
Applying the above definitions, the proof of the next lemma is immediate.
\begin{Lemma}
 The code $C$ is   right $(\sigma, \delta)$-sequential  if and only if it is invariant under $T_{C_f^t}=\sigma(g) C_f^t+\delta(g)$.  Moreover, the  code $C$ is  left $(\sigma, \delta)$-sequential if and only if it is invariant under $T_{E_d^t}=\sigma(g) E_d^t+\delta(g).$ Here,  $t$ denotes  the transpose matrix.
     \end{Lemma}

Recall that an  anti-automorphism $\zeta: R\to R^{op}$ is a map which sends $ab$ to $\zeta(b)\zeta(a),$ where  $R^{op}$ denotes the opposite ring of $R.$
A $\zeta$-sesquilinear form on an $R$-module $M$ associated  with an  anti-automorphism $\zeta$ is given by $\langle \cdot , \cdot \rangle : M\times M\to R$
that satisfies the following conditions: for all $x, y, z$ in $M$ and all $\alpha, \beta\in R$,  $\langle  x+y,z\rangle =\langle x,z \rangle+\langle y,z \rangle,$
$\langle  x,y+z\rangle = \langle x,y \rangle+\langle x,z \rangle,$  $\langle \alpha x , y \rangle= \alpha\langle x,y \rangle$ and $\langle x , \beta y \rangle= \langle x,y \rangle \zeta(\beta).$  
For an anti-automorphism  $\zeta$ on the ring  $R$, define the map 
$\langle\cdot, \cdot\rangle : R^n \times R^n \rightarrow R$  by   
\begin{equation}\label{501}
    \langle g, h\rangle = \sum_{i=0}^{n-1} g_i\zeta(h_i),
\end{equation} 
 where $g = (g_0, g_1, \ldots, g_{n-1})$ and  $h = (h_0, h_1, \ldots, h_{n-1}).$ This map is  a $\zeta$-sesquilinear form, see \cite{form} Proposition 3.1. For any code $C \subseteq R^n$, either linear or not, 
its left  dual code (denoted by $l(C)$)  and its right  dual code (denoted by $r(C)$)  are defined as 
\begin{equation}\label{1745}
    l(C)=\big\{g\in R^n: \langle g ,h\rangle=0\,\, \text{for all}\,\, h\in C \big\} \quad \text{and} \quad 
 r(C)=\big\{g\in R^n: \langle h ,g\rangle=0\,\, \text{for all}\,\, h\in C \big\}. 
\end{equation}
In the case $\zeta(\langle x, y \rangle) = \langle y, x \rangle$, it follows that  $l(C) = r(C)$.
In the following theorem, Szabo and Wood have presented the double-dual and cardinality complementary properties of codes over finite Frobenius rings for non-degenerate sesquilinear forms.
\begin{Theorem}[\cite{form} Theorem 6.2]\label{603}
Consider a left $R$-linear code $C \subseteq R^n$ over a finite Frobenius ring $R$.  Assume $R^n$ is equipped with a non-degenerate sesquilinear form. Then
\begin{enumerate}
       \item  $l(C)$ and $r(C)$ are left $R$-linear codes. 
          \item $l(r(C)) = C$ and $r(l(C)) = C$.
    \item $|C| \cdot |l(C)| = |R^n|$ and $|C| \cdot |r(C)| = |R^n|$.
\end{enumerate}
\end{Theorem}

From now one,  assume that   the  anti-automorphism $\zeta$ is the identity operator  $\operatorname{Id}:R\to R^{op},$ which maps  $ab$ to $ba$. Therefore, $\zeta$-sesquilinear form $\langle g, h\rangle = \sum_{i=0}^{n-1} g_i\zeta(h_i)$ is equal to $\sum_{i=0}^{n-1} g_ih_i=gh^t$, where $t$ denotes the transpose. We have,  $\langle h, g \rangle = \sum_{i=1}^{n} h_ig_i=\sum_{i=1}^{n}\zeta(g_ih_i) = \sum_{i=1}^{n} g_ih_i = \langle g, h \rangle$. Moreover, $\langle g, \beta h \rangle = \langle g, h \rangle \zeta(\beta) = \zeta(\langle g, h \rangle) \zeta(\beta)= \zeta (\beta(\langle g, h \rangle))= \beta\langle g, h \rangle.$
Therefore,  in the case $\zeta=\operatorname{Id}$, the $\zeta$-sesquilinear form  $\langle \cdot \,, \cdot \rangle$ is exactly the standard inner product and $l(C)=r(C)=C^{\perp}$.

Before establishing the relationship between $(\sigma, \delta)$-polycyclic codes and their duals, we need the following easy lemma.
\begin{Lemma}\label{530} If   $\sigma\in \operatorname{Aut}(R)$, then $\delta_1:=-\sigma^{-1}\delta$ is a left $\sigma^{-1}$-derivation on $R.$ In other words,
  it  is a  right $\sigma^{-1}$-derivation on the opposite ring $R^{op}$.
\end{Lemma}
\begin{proof}
For any $a, b \in R$, we have
\[
\delta_1(ab) = -\sigma^{-1} \big( \sigma(a)\delta(b) + \delta(a)b \big)
= -a\sigma^{-1}\delta(b) - \sigma^{-1}\delta(a)\sigma^{-1}(b)
=  a\delta_1(b)+\delta_1(a)\sigma^{-1}(b).
\]
\end{proof}
\begin{Theorem}\label{308}
   Assume  $\sigma\in \operatorname{Aut}(R)$. If  $C$ is a right $(\sigma, \delta)$-polycyclic code in $S/Sf,$ then $C^{\perp}$
is invariant under $T_{C_{\sigma'(f)}^t}(v)=\sigma'(v)C_{\sigma'(f)}^t+\delta'(v)$, where $\sigma'=\sigma^{-1}$, $\delta'=-\sigma^{-1}\delta$.  In fact, $C^{\perp}$
is a right $(\sigma',\delta')$-sequential  code in $S'/S'f'$, where $S'$ is the Ore extension  $R^{op}[x, \sigma', \delta']$ and  $f'$ is obtained by applying $\sigma'$ to the coefficients of $f$.

\end{Theorem}
\begin{proof}
Assume $C$   is a right $(\sigma, \delta)$-polycyclic code and  $g\in C$ is arbitrary. Since $T_{C_f}(g)\in C,$ for $h\in C^{\perp}$ we have
  \begin{equation}\label{90}
      0=\langle \sigma(g)C_f+\delta(g),\,h\rangle=\langle \sigma(g)C_f,\,h\rangle+\langle \delta(g),\,h\rangle.
  \end{equation}
On the other side, since  $0=\langle g,h\rangle=gh^t$, we obtain $0=\delta(gh^t)=\sigma(g)\delta(h^t)+\delta(g)h^t=\sigma(g)\delta(h)^t+\delta(g)h^t$. Hence $-  \sigma(g)\delta(h)^t=\delta(g)h^t,$ that is $\langle \sigma(g), -\delta(h)\rangle =\langle \delta(g), h \rangle .$ By substituting the last equality in Equation \eqref{90} and using the fact that 
$$\langle \sigma(g) C_f,\,h\rangle=\sigma(g_{n-1})f_0h_0+\sigma(g_{n-1})f_1h_1+\sigma(g_0)h_1+\cdots+\sigma(g_{n-1})f_{n-1}h_{n-1}+\sigma(g_{n-2})h_{n-1}=\langle \sigma(g),\,hC_f^t\rangle,$$
we get 
\begin{align*}
    0 &= \langle \sigma(g),\,h C_f^t\rangle+\langle \sigma(g),\,-\delta(h)\rangle\\
    &=\langle \sigma(g),\,hC_f^t-\delta(h)\rangle\\
    &=\langle hC_f^t-\delta(h),\, \sigma(g)\rangle\\
    &= \big(hC_f^t-\delta(h)\big)\sigma(g^t).
    \end{align*}
    So,  $\bigg(\sigma^{-1}(h)\sigma^{-1}(C_f^t))-\sigma^{-1}\delta(h)\bigg)g^t=0$. It can easily be seen  that $\sigma^{-1}(C_f^t))=C_{\sigma^{-1}(f)}^t, $
     where $\sigma^{-1}(f)$  is obtained by applying $\sigma^{-1}$ to the coefficients of $f$.
Thus,  $\langle g,\,\,\sigma^{-1}(h)C_{\sigma^{-1}(f)}^t-\sigma^{-1}\delta(h)\rangle=0.$ Let us define $$T_{C_{\sigma^{-1}(f)}^t}(v):=\sigma^{-1}(v)C_{\sigma^{-1}(f)}^t-\sigma^{-1}\delta(v).$$ 
For any given $g$, we have proved that $\langle g, T_{C_{\sigma^{-1}(f)}}(h) \rangle=0$, that is
$ T_{C_{\sigma^{-1}(f)}}(h)\in C^{\perp}.$
To sum up, for any given $h\in C^{\perp},$ we get $T_{C_{\sigma^{-1}(f)}^t}(h)\in C^{\perp},$ which implies that
 $C^{\perp}$ is a  $(\sigma^{-1}, -\sigma^{-1}\delta )$-sequential code  in $R^{op}[x, \sigma^{-1}, -\sigma^{-1}\delta]/R^{op}[x, \sigma^{-1}, -\sigma^{-1}\delta]{\sigma}^{-1}(f)$. 
\end{proof}


 \section{Duality}\label{S3}
In Section \ref{S2}, we observed that the Euclidean dual of a $(\sigma,\delta)$-polycyclic code may not necessarily be a $(\sigma,\delta)$-polycyclic code. In the current section, we aim to introduce an alternative duality that ensures the dual of any $(\sigma,\delta)$-polycyclic code remains a $(\sigma,\delta)$-polycyclic code.

\begin{Lemma}\label{410}
     Suppose that  $f_0$ is invertible and  $\sigma(f_i)=f_i$  and   $\delta(f_i)=0$ for all $i$ satisfying $0\leqslant i\leqslant n-1.$ Then  $x\in S/Sf$ is a right invertible element.
\end{Lemma}
\begin{proof}
 Define $k(x)=\sum_{i=0}^{n-2}-f_0^{-1}f_{i+1}x^i+f_0^{-1}x^{n-1}.$  We get $0=\delta(f_0^{-1}f_0)=\sigma(f_0^{-1})\delta(f_0)+\delta(f_0^{-1})f_0,$ and thus $\delta(f_0^{-1})f_0=0.$ This implies that $\delta(f_0^{-1})=0.$ Moreover,  we obtain $\delta(f_0^{-1}f_i)=\sigma(f_0^{-1})\delta(f_i)+\delta(f_0^{-1})f_i=0.$  Therefore, 
\begin{align}\label{470}
    xk(x)&= \sum_{i=0}^{n-2}x\bigg(-f_0^{-1}f_{i+1}\bigg)x^i+x\bigg(f_0^{-1}\bigg)x^{n-1}\\\
    &=-\sum_{i=0}^{n-2} \bigg(\sigma(f_0^{-1}f_{i+1})x+\delta(f_0^{-1}f_{i+1})\bigg)x^i+\bigg(\sigma(f_0^{-1})x+\delta(f_0^{-1})\bigg)x^{n-1}\nonumber\\
    &=-\sigma(f_0^{-1})\bigg(\sum_{i=0}^{n-2} f_{i+1} x^{i+1}+x^{n}\bigg)\nonumber\\
    &=-\sigma(f_0^{-1})\bigg(\sum_{i=0}^{n-2} f_{i+1} x^{i+1}+x^{n}\bigg)-\sigma(f_0^{-1}f_0)+\sigma(f_0^{-1}f_0)\nonumber\\
    &=\sigma(f_0^{-1})f(x)+1= 1\,( \operatorname{mod} f).\nonumber
\end{align}
\end{proof}

\begin{Definition}
     Define the map $\langle \cdot \,, \cdot \rangle_0 : S/Sf\times S/Sf\to R$ by $$\langle g , h\rangle_0=gh(0).$$
\end{Definition}
Recall that   $T_0 = \delta$ is a $(\sigma, \delta)$-PLT and $\delta(1)=0.$ Let  $g,\,h\in S/Sf$ and  $gh\,(\operatorname{mod} f)$ be denoted by the polynomial $r(x)=\sum_{i=0}^{n-1} r_ix^i$. Note that $r(0)=\sum_{i=0}^{n-1} r_iN_i(0)=r_0$, so $gh(T_0)(1)=r_0=gh(0)$. Therefore, we have the equivalent definition  $\langle g , h\rangle_0=gh(T_0)(1)$, which provides technical assistance in our proofs.
\begin{Lemma}\label{701}
Suppose that  $f_0$ is invertible and  $\sigma(f_i)=f_i$  and   $\delta(f_i)=0$ for all $i$ satisfying $0\leqslant i\leqslant n-1.$   Then $\langle \cdot , \cdot \rangle_0$   is a non-degenerate $\zeta$-sesquilinear form, where $\zeta:=\operatorname{Id}$.
\end{Lemma}
\begin{proof}
Let $g,h\in S/Sf$. For a fixed  $\beta \in R$ and  for any $r\in R,$ we have 
$$\beta h(T_0)(r)=\operatorname{Id}\bigg(\beta h(T_0)(r)\bigg)=h(T_0)(r)\beta=h(T_0)(r)\operatorname{Id}(\beta)=h(T_0)(r)\zeta(\beta).$$
We thus have $g\beta h(T_0)(r)=gh(T_0)(r)\zeta(\beta),$ which implies  $\langle g, \beta h\rangle_0=\langle g,h\rangle_0\zeta(\beta).$ Verifying the other conditions of the $\zeta$-sesquilinear form definition is easy. Therefore,  $\langle  \cdot, \cdot\rangle_0$  is a $\zeta$-sesquilinear form. To prove the non-degeneracy of the form, suppose that  $g = g_0 + g_1 x + \cdots + g_{n-1} x^{n-1} \in S/Sf$ satisfies $\langle g , h \rangle_{0} = 0$ for all $h \in S/Sf.$
Let us assume   $h=1. $  This gives  $0=g(T_0)(1)=g(\delta)(1)=g_0.$  Consequently, we can rewrite $g=(g_1+g_2x+\cdots+g_{n-1}x^{n-2})x$. By Lemma \ref{410}, $x$ is right invertible. Let us denote its inverse by $x^{-1}$. If we  assume $h=x^{-1}$, we get $0=\langle g,h\rangle_0=gh(\delta)(1)=g_1$. 
By applying the same discussion, we deduce that $g_i = 0$ for all $i$. Consequently, we have $g = 0$, implying that the form is non-degenerate.
      \end{proof}
\begin{Definition}
  For the right  $(\sigma,\delta)$-polycyclic code  $C$, define the left annihilator dual by 
$$l_0(C)=\{g\in S/Sf: \langle g,h\rangle_0=0,\,\, \text{for all}\,\, h\in C\}=\{g\in S/Sf: \langle g,C\rangle_0=0\}. $$  
Furthermore, for the left  $(\sigma,\delta)$-polycyclic code  $C$, define the right annihilator dual by 
$$r_0(C)=\{g\in S/Sf: \langle h ,g\rangle_0=0,\,\, \text{for all}\,\, h\in C\}=\{g\in S/Sf: \langle C, g\rangle_0=0\}. $$  
\end{Definition}
\begin{Remark}
 By Lemma \ref{20}, we have $ \langle g,h\rangle_0=gh(0)=g(T_0)(h(0)),$ which is not equal to $h(T_0)(g(0))=\langle h, g\rangle_0$ in generally.  Therefore,  $l_0(C)$ and $r_0(C)$ are not equal in generally.
   Assume  that $R$ is commutative and  $(\sigma,\delta)=(\operatorname{Id},0).$  In this case, $T_0 = 0$, and  $\langle g,h\rangle_0 = \langle h, g\rangle_0$. Therefore  $l_0(C)=r_0(C)$ and both are equal to the annihilator dual $C^{\perp_0}$ as  defined in \cite{Alahmadi, our}.
\end{Remark}
Referring to \cite{noncommutativebook}, for a subset $X$  in a non-commutative ring $A$, the left annihilator $\operatorname{Ann}_l(X)$ and the right annihilator $\operatorname{Ann}_r(X)$ of  $X$ are defined by  $$\operatorname{Ann}_l(X) = \{a \in A :  ax = 0,\,\, \text{for all}\,\, x \in X\}\quad \text{and}\quad \operatorname{Ann}_r(X) = \{a \in A :  xa = 0,\,\, \text{for all}\,\, x \in X\}.$$
\begin{Theorem}\label{520}
   Suppose that  $f_0$ is invertible and  $\sigma(f_i)=f_i$  and   $\delta(f_i)=0$ for all $i$ satisfying $0\leqslant i\leqslant n-1.$ 
   If $C$ is  a right $(\sigma,\delta)$-polycyclic code,   then  $l_0(C)=\operatorname{Ann}_l(C)$. In addition, if $C$ is  a left $(\sigma,\delta)$-polycyclic code, then
  $r_0(C)=\operatorname{Ann}_r(C).$ 
\end{Theorem}
\begin{proof}
We have $\operatorname{Ann}_l(C)\subseteq l_0(C).$ Conversely, let $g\in l_0(C)$, $h$ be an arbitrary element of $C$, and  $gh\,(\operatorname{mod} f)$ be denoted by the polynomial $r_0 + r_1x + \cdots + r_{n-1}x^{n-1}$.
We have  $gh(0)=0,$ which  leads to  $r_0 = 0 $. So, $ gh(x)=(r_1+r_2x+\cdots+r_{n-1}x^{n-2})x.$ Using Lemma \ref{410},
\begin{equation}\label{460}
  ghk(x)=(r_1+r_2x+\cdots+r_{n-1}x^{n-2})\,\,(\operatorname{mod} f).
\end{equation}

On the other side,  we have $\delta(k(0))=\delta(-f_0^{-1}f_1)=-\sigma(f_0^{-1})\delta(f_1)-\delta(f_0^{-1})f_1$. Since $\delta(f_1)=\delta(f_0^{-1})=0$, we get $\delta(k(0))=0.$
By Lemma \ref{20}, $ghk(0)=gh(T_0)(k(0))=gh(\delta)(k(0))=gh(0)=0$. Therefore $r_1=0$ by Equation \eqref{460}. Applying the same discussion, we can deduce that $r_i = 0$ for all $i$. Consequently, $gh = 0$,  which implies to $g\in \operatorname{Ann}_l(C).$ Therefore, $l_0(C)=\operatorname{Ann}_l(C)$. The proof of the second equality is similar.

\end{proof}
\begin{Corollary}\label{600}
 Suppose that  $f_0$ is invertible and  $\sigma(f_i)=f_i$  and   $\delta(f_i)=0$ for all $i$ satisfying $0\leqslant i\leqslant n-1.$
 \begin{enumerate}
     \item  If $C$ is a right $(\sigma,\delta)$-polycyclic code then  $l_0(C)$ is  a 
 right $(\sigma,\delta)$-polycyclic code.
 \item If $C$ is a left $(\sigma,\delta)$-polycyclic code then  $r_0(C)$ is  a 
 left $(\sigma,\delta)$-polycyclic code.
 \end{enumerate}
\end{Corollary}
\begin{proof}
Assume that   $C$ is a right $(\sigma,\delta)$-polycyclic code. Let $a\in \operatorname{Ann}_l(C)$ and  $s\in S$. For any given arbitrary element $c\in C$, we know that $sc$ also belongs to $C$. Hence,
$(as)c=a(sc)=0,$  which implies $as\in \operatorname{Ann}_l(C).$ Thus  $\operatorname{Ann}_l(C)$ is a left $S$-submodule of $S/Sf.$ Using  Theorem \ref{520}, $l_0(C)$ is a left $S$-submodule of $S/Sf,$ which means $l_0(C)$ is a right $(\sigma,\delta)$-polycyclic code. The proof of the second statement is similar.
\end{proof}

\section{Equivalent $(\sigma,\delta)$-polycyclic codes}\label{S4}
The Hamming weight $\omega_H(g)$
of a vector $g\in R^n$
 is the number of nonzero coordinates in $g$. A linear map $\psi: C_1 \rightarrow C_2$ is called  \textit{Hamming isometry} if it is an isomorphism that preserves the Hamming weight, meaning that for every codeword $c \in C_1$, we have $\omega_H(\psi(c)) = \omega_H(c)$. Two codes $C_1$ and $C_2$ are called  \textit{Hamming isometrically equivalent} if there exists a Hamming isometry $\psi: C_1 \rightarrow C_2$.

\begin{Proposition}[\cite{LEROY} Proposition 1.6]\label{180}
Let $T_1$ and $T_2$ be  $(\sigma, \delta)$-PLTs defined on  free $R$-modules $V_1$ and $V_2$ with basis $\mathcal B_1$ and $\mathcal B_2,$ respectively. Suppose $\varphi \in \operatorname{Hom}_{R}(V_1, V_2)$ is an $R$-module homomorphism, and let  $B $, $M_1 $ and $M_2 $ denote matrix representations of $\varphi$, $T_1$ and $T_2,$ respectively.  The following conditions are equivalent.
\begin{enumerate}
    \item $\varphi \in \operatorname{Hom}_{R}(V_1, V_2)$.
    \item $\varphi T_1 = T_2 \varphi$.
    \item $M_1B = \sigma(B)M_2 + \delta(B)$.
    \end{enumerate}
\end{Proposition}
Leroy's paper \cite{LEROY} includes the following corollary without proof. In order to provide context for the next theorem, we will present its proof.
\begin{Corollary}[\cite{LEROY} Corollary 1.8]\label{187} 
Consider the same notations as the previous proposition. Let $f_1, f_2 \in S$ be two monic polynomials of degree $n$ with companion matrices $C_{f_1}, C_{f_2}$. The map  $\varphi: S/Sf_1 \to S/Sf_2$   is an $R$-module isomorphism if and only if $B$ is an invertible matrix and
$C_{f_1}B = \sigma(B)C_{f_2} + \delta(B)$.
\end{Corollary}
\begin{proof}
    Define the linear transformation $\psi_B: R^n \to R^n$ by $\psi_B(v) = vB.$ Assume that  $B$ is an invertible matrix and  $C_{f_1}B = \sigma(B)C_{f_2} + \delta(B)$.  Proposition \ref{180} implies that  $\varphi: S/Sf_1\to S/Sf_2$ is an $R$-module homomorphism. Now, consider the following commutative diagram. 
\begin{center}
$$\begin{tikzcd}
S/Sf_1 \arrow[r, "\varphi"] \arrow[d, swap, "\Omega_{f_1}^{-1}"] & S/Sf_2 \arrow[d, "\Omega_{f_2}^{-1}"] \\
R^n \arrow[r, "\psi_B"] & R^n
\end{tikzcd}$$
\vspace{-0.5cm}
\captionof{figure}{}
\label{1900}
\end{center}
The map  $\psi_B$ is injective, and due to the equal cardinality of its domain and codomain,  it is an $R$-module isomorphism.  Additionally, since $\Omega_{f_1}$ and $\Omega_{f_2}$ are both isomorphisms, it follows that $\varphi$ is also an isomorphism. The converse can be proved using the same discussion.
\end{proof}

\begin{Corollary}\label{705}
 $\psi_B$ is an $R$-module isomorphism if and only if   $B$ is an invertible and
$C_{f_1}B = \sigma(B)C_{f_2} + \delta(B)$.
\end{Corollary}
\begin{Theorem}\label{708}
Suppose that $B$ is a monomial matrix and
$C_{f_1}B = \sigma(B)C_{f_2} + \delta(B)$. Then, there is a one-to-one Hamming isometrically equivalent correspondence between the set of right $(\sigma,\delta)$-polycyclic codes in $S/Sf_1$ and right $(\sigma,\delta)$-polycyclic codes in $S/Sf_2.$   
\end{Theorem}
\begin{proof}\

\textbf{Claim 1}: \textit{$\varphi$ is a Hamming isometry if and only if $\psi_B$ is a Hamming isometry.} To prove the claim, note that for every  $g_1(x)\in S/Sf_1$ (which associated with $g_1\in R^n$),  and for every  $g_2(x)\in S/Sf_2$ (which associated with $g_2\in R^n$), we have $\omega_H(g_1(x))=\omega_H(g_1)$ and $\omega_H(g_2(x))=\omega_H(g_2).$ Let $\varphi$ be a   Hamming isometry. From the commutativity of the diagram in Figure \ref{1900}, we have $\Omega_{f_2}^{-1}\varphi=\psi_B\Omega_{f_1}^{-1}.$ Therefore, $\omega_H(g_1)=\omega_H(g_1(x))=\omega_H(\varphi(g_1(x)))=\omega_H(\Omega^{-1}_{f_2}\varphi(g_1(x)))=\omega_H(\psi_B\Omega^{-1}_{f_1}(g_1(x)))=\omega_H(\psi_B(g_1)).$ Thus, $\psi_B$ is a Hamming isometry. The proof of the converse is straightforward.

\textbf{Claim 2}: \textit{$\psi_B$ is a Hamming isometry if and only if $B$ is a monomial matrix.} The proof from right to left is clear and the proof of another direction derives from  Proposition 6.1 in \cite{wood}.

To prove the theorem, let $C_1$  be a right $(\sigma,\delta)$-polycyclic code in $S/Sf_1$. Corollaries \ref{187}, \ref{705} and Claims 1, 2  imply that the restriction of $\psi_B$ to $C_1$, denoted as $\psi_B: C_1 \to C_2:=\psi_B(C_1)$, is a Hamming isometry.
\end{proof}
As an application of the above theorem, we have the next theorem, which says that both left and right $(\sigma,\delta)$-polycyclic codes lead to equivalent theories. For this reason,  we have only focused on the right $(\sigma,\delta)$-polycyclic codes.
\vspace{0.4cm}

\textbf{In the sequel of this paper, assume that $R$ is the field $\mathbb F_q$ and $\sigma\in \operatorname{Aut}(R)$. }

\section{$(\sigma,\delta)$-Mattson-Solomon transform for simple-root $(\sigma,\delta)$-polycyclic codes} \label{26}
The theory of Wedderburn polynomials in  Ore extension rings, extensively explored by Lam and Leroy in \cite{lamleroy1, lamleroy2}, forms a foundational basis for two next sections.
In this section, our first objective is to define simple-root $(\sigma,\delta)$-polycyclic codes, using the concept of  Wedderburn polynomials in the Ore extension $S=\mathbb F_q[x,\sigma,\delta]$.  Our second objective is to present a definition for $(\sigma,\delta)$-Mattson-Solomon transform for simple-root $(\sigma,\delta)$-polycyclic codes.
 
\subsection{Algebraic sets and minimal polynomials}\label{123}
For any  polynomial $g\in S$, define $V(g):=\{a\in \mathbb F_q :  g(a)=0\},$ and  for any  subset $X \subseteq \mathbb F_q$, define  $I(X):=\{g\in S : g(X)=0 \}.$ The  set $X$ is called  $(\sigma,\delta)$-\textit{algebraic}  if $I(X)\neq 0$. The set $V(g)$ is $(\sigma,\delta)$-algebraic because $g(V(g))=0$. Moreover, any finite set is $(\sigma,\delta)$-algebraic, see \cite{lamleroy4} Section $6$. If $X$ is $(\sigma,\delta)$-algebraic, then the monic generator of $I(X)$ is called the minimal polynomial of $X$ and is denoted by $g_{X}.$ Clearly $g_{X}(X)=0$. By the Remainder Theorem,  the minimal polynomial $g_X$ is the monic least right common multiple of the linear polynomials $\{x - a: a \in X\}$. So it has  the form $(x-a_0)(x-a_1) \cdots (x-a_{n-1}),$ where each $a_i$ is $(\sigma, \delta)$-conjugate to some element of $X$ and $n=\operatorname{deg}g_{X}$.  Furthermore, any right root of $g_{X}$ is also $(\sigma, \delta)$-conjugate to some element of $X$, see \cite{lamleroy3} Section 4 and \cite{lamleroy1} Section 2. 
An element $a\in \mathbb F_q$ is called \textit{P-dependent} on an algebraic set $X$ if $g(a)=0$ for every $g \in I(X)$. A $(\sigma,\delta)$-algebraic set $X$ is P-independent if no element $b \in X$ is P-dependent on $X \setminus {b}$, see \cite{lamleroy1}.

\subsection{Wedderburn polynomials}\label{news}
A monic polynomial $g \in S$ is called  \textit{Wedderburn polynomial} (or \textit{W-polynomial} for short) if the minimal polynomial of $V(g)$ is $g.$ If $(\sigma,\delta)=(\operatorname{Id}, 0)$, then  W-polynomials are in the form of $(x-a_0)\cdots(x-a_{n-1})$, where $a_i$s are distinct elements in $\mathbb F_q$, see \cite{lamleroy1} Example 3.5.

Recall from Section \ref{590} that if $g(x)=\sum_{i=0}
^{n-1} g_ix^i\in S$,  then $g(a)=\sum_{i=0}
^{n-1} g_iN_i(a).$

The Vandermonde matrix is a valuable tool for polynomial evaluation. Hence, in  \cite{vandermond}, the Vandermonde matrix in the $(\sigma,\delta)$-setting is naturally defined as follows. A $(\sigma,\delta)$-\textit{Vandermonde matrix} related to the elements $a_0,\ldots,a_{n-1} \in \mathbb F_q ,$  denoted by $V(a_0, a_1,\ldots,a_{n-1})$ or simply by $V$,  is 
    $$ V=V(a_0, a_1,\ldots,a_{n-1})=\begin{pmatrix}
1 & 1 & \cdots & 1 \\
N_0(a_0) & N_0(a_1) & \cdots & N_0(a_{n-1}) \\
\vdots & \vdots & \cdots & \vdots \\
N_{n-1}(a_0) & N_{n-1}(a_1) & \cdots & N_{n-1}(a_{n-1}) 
\end{pmatrix}.  $$

The following theorem, which plays a crucial role in both this section and the subsequent one, presents some equivalent conditions for W-polynomials.

\begin{Theorem}[\cite{lamleroy2} Theorem 5.2, Lemma 5.7 and Theorem 5.8]\label{4400}
Let  $f \in S$ be a  monic polynomial of degree $n$, then the following statements are equivalent.
\begin{enumerate}[i)]
\item $f$ is a W-polynomial.
\item There is a P-independent set $A = \{a_0, a_1, \ldots, a_{n-1}\} \subseteq \mathbb F_q$ such that $f = f_A$.
\item There is a set  $\{a_0, a_1, \ldots, a_{n-1}\} \subseteq \mathbb F_q$ such that the $(\sigma, \delta)$-Vandermond matrix  $V = V(a_0, a_1, \ldots, a_{n-1})$ is invertible and $$C_fV = \sigma(V) \operatorname{diag} (a_0, a_1, \ldots, a_{n-1}) + \delta(V).$$
\item  The companion matrix  $C_f$ is $(\sigma,\delta)$-diagonalizable.
\item $Sf= \bigcap_{i=0}^{n-1} S(x-a_i)$.
\item The left $S$-module $S/Sf$ is decomposed as $\bigoplus_{i=0}^{n-1} S/S{(x-a_i)}$ .
\end{enumerate}
\end{Theorem}

\subsection{$(\sigma,\delta)$-Mattson-Solomon transform}

\begin{Definition}
A $(\sigma,\delta)$-polycyclic code in $S/Sf$ is called  simple-root $(\sigma,\delta)$-polycyclic code if $f$ is a W-polynomial.
\end{Definition}
A polycyclic code in the quotient ring $\mathbb F_q[x]/\langle f(x)\rangle$ is called simple-root if $f(x)$ is a simple-root polynomial in $\mathbb F_q[x]$, see \cite{our}.
In the case where $(\sigma,\delta)=(\operatorname{Id}, 0)$, W-polynomials are precisely  simple-root polynomials. Therefore,  simple-root $(\sigma,\delta)$-polycyclic codes are the generalization of simple-root polycyclic codes.

Let $f$ be a W-polynomial with a P-independent set $\{a_0, a_1, \ldots, a_{n-1}\}$ . Define  $T: S/Sf\to S/Sf$   by $T(\sum_{i=0}^{n-1} b_{i}x^{i})=\sum_{i=0}^{n-1} T_{a_i}(b_i)x^i$, where $T_{a_i}: R\to R$ is the $(\sigma,\delta)$-PLT given by $T_{a_i}(x)=\sigma(x)a_i+\delta(x)$. 

\begin{Lemma}\label{60005}
 The additive map $T$ is a $(\sigma,\delta)$-PLT.
\end{Lemma}
\begin{proof}
For all $a\in \mathbb F_q$ and $v=\sum_{i=0}^{n-1} b_ix^i\in S/Sf$ we get
\begin{align*}
    T(av)&= T\bigg(\sum_{i=0}^{n-1} ab_ix^i\bigg)=\sum_{i=0}^{n-1} T_{a_i}(ab_i)x^i\\
    &= \sum_{i=0}^{n-1}\bigg(\sigma(ab_i)a_i+\delta(ab_i)\bigg)x^i\\
    &=\sum_{i=0}^{n-1}\bigg(\sigma(a)\sigma(b_i)a_i+\sigma(a)\delta(b_i)+\delta(a)b_i\bigg)x^i\\
     &=\sigma(a)\sum_{i=0}^{n-1}\bigg(\sigma(b_i)a_i+\delta(b_i)\bigg)x^i+\delta(a)\sum_{i=0}^{n-1} b_ix^i=\sigma(a)T(v)+\delta(a)V.
\end{align*}
\end{proof}
 
For all  $v=\sum_{i=0}^{n-1} b_ix^i\in S/Sf$ and for all $k\in \mathbb N,$ we have $T^k(\sum_{i=0}^{n-1} b_ix^i)=\sum_{i=0}^{n-1}  T_{a_i}^k(b_i)x^i$. Therefore
\begin{equation}\label{60003}
   g(T)v=g(T)\bigg(\sum_{i=0}^{n-1} b_ix^i\bigg)=\sum_{i=0}^{n-1} g(T_{a_i})(b_i)x^i, \quad \text{for all} \,\, g\in S/Sf.
\end{equation}
Using Lemma \ref{60005} and Proposition \ref{5}, the  $\mathbb F_q$-module $S/Sf$  has the $S$-module structure induced by the product $g(x)\cdot v:=g(T)v$. Let us denote this structure by $(S/Sf, \cdot)$.
Furthermore, we know $S/Sf$  also possesses the $S$-module structure induced by the product $g(x) \bullet h(x) = g(x)h(x)\, (\operatorname{mod} f)$. This structure will be denoted by $(S/Sf, \bullet)$.
\begin{Definition}\label{60}
For a W-polynomial $f$  with an associated  P-independent set $A = \{a_0, a_1, \ldots, a_{n-1}\}$, define a $(\sigma,\delta)$-Mattson-Solomon transform as  
\begin{align*}
  \operatorname{MS}_{(\sigma,\delta)}: (S/Sf, \bullet)&\longrightarrow (S/Sf, \cdot)\\
   g(x) &\mapsto \sum_{i=0}^{n-1}g(a_i)x^i.
\end{align*}   
\end{Definition}
\begin{Theorem}
The $(\sigma,\delta)$-Mattson-Solomon transform $\operatorname{MS}_{(\sigma,\delta)}$ is a $S$-module isomorphism. 
\end{Theorem}
\begin{proof}
Recall that if   $g(x)=\sum_{i=0}^{n-1} g_i x^i \in S$, then $g(a)=\sum_{i=0}^{n-1} g_iN_i(a)$. If we associate the polynomial $g(x)$  with the vector $g=(g_0,g_1,\ldots,g_{n-1}),$ then $gV=(g(a_0), g(a_1),\ldots, g(a_{n-1}))$, where $V$ is the $(\sigma,\delta)$-Vandermond matrix related to $a_0, a_1, \cdots, a_{n-1}.$ Hence, as depicted in the Figure \ref{6000}, the vector form of $\operatorname{MS}_{(\sigma,\delta)}$ is given by $\operatorname{MS}_{(\sigma,\delta)}(g) = gV$.
\vspace{-0.6cm}
\begin{center}
$$\begin{tikzcd}
S/Sf \arrow[r, " \operatorname{MS}_{(\sigma,\delta)}"] \arrow[d, swap, "\Omega_{f}^{-1}"] & S/Sf \arrow[d, "\Omega_{f}^{-1}"] \\
R^n \arrow[r, "\operatorname{MS}_{(\sigma,\delta)}(g)=gV"] & R^n
\end{tikzcd}$$
\vspace{-0.5cm}
\captionof{figure}{}
\label{6000}
\end{center}
Using Theorem \ref{4400}, $V$ is an invertible matrix,  and hence $\operatorname{MS}_{(\sigma,\delta)}$ is a bijection. Furthermore, $\operatorname{MS}_{(\sigma,\delta)}$ is a $S$-module homomorphism because
\begin{align*}
    \operatorname{MS}_{(\sigma,\delta)}(g(x)h(x))&= \sum_{i=0}^{n-1} gh(a_i)x^i\\
    &=  \sum_{i=0}^{n-1} g(T_{a_i})h(a_i)x^i\qquad \text{By Lemma \ref{20} }\\
    &= g(T)\bigg( \sum_{i=0}^{n-1} h(a_i)x^i\bigg)\qquad \text{By Equation \eqref{60003}}\\
    &= g(T) \bigg(\operatorname{MS_{(\sigma,\delta)}(h(x))}\bigg)=g(x)\cdot \operatorname{MS}_{(\sigma,\delta)}(h(x)).
\end{align*}
\end{proof}
\begin{Remark}
Let us equip $S/Sf$ with the component-wise product $(\sum_{i=0}^{n-1} b_ix^i)\star (\sum_{i=0}^{n-1} c_ix^i)=\sum_{i=0}^{n-1} b_ic_ix^i$.  Now, suppose that   $(\sigma,\delta)=(\operatorname{Id},0)$. Then 
\begin{align*}
   g(x)\cdot \sum_{i=0}^{n-1} h(a_i)x^i &= g(T)\bigg(\sum_{i=0}^{n-1} h(a_i)x^i\bigg)\\
    &= \sum_{i=0}^{n-1} g(T_{a_i})(h(a_i)) x^i \qquad \text{Equation \eqref{60003}}\\
    &= \sum_{i=0}^{n-1} gh(a_i)x^i \qquad \text{Lemma \ref{20} }\\
    &= \sum_{i=0}^{n-1} g(a_i)h(a_i)x^i=\sum_{i=0}^{n-1} g(a_i)x^i\star \sum_{i=0}^{{n-1}} h(a_i)x^i.
    \end{align*}
So, due to the bijectivity of  $\operatorname{MS}_{(\sigma,\delta)}$,  the two operations $\cdot$ and $\star$ are equivalent. Furthermore,  $$\operatorname{MS}_{(\sigma,\delta)}(g(x)h(x))=\operatorname{MS}_{(\sigma,\delta)}(g(x))\star \operatorname{MS}_{(\sigma,\delta)}(h(x)).$$ Thus,
   $\operatorname{MS}_{(\sigma,\delta)}$ is a ring homomorphism. Therefore, in the case where $(\sigma,\delta)=(\operatorname{Id},0)$, the $(\sigma,\delta)$-Mattson Solomon map  $\operatorname{MS}_{(\sigma,\delta)}:(S/Sf, \bullet) \to (S/Sf, \star)$ corresponds to the Mattson Solomon map of simple-root polycyclic codes defined in \cite{our}.
\end{Remark}
\section{Decomposition of simple-root $(\sigma,\delta)$-polycyclic codes}\label{MM5}
Let us recall that $S=\mathbb F_q[x,\sigma,\delta]$, and consider $C$ as a $(\sigma,\delta)$-polycyclic code. In this section, we aim to find the direct sum decomposition $C=C_0\oplus C_1\oplus\cdots \oplus C_{n-1},$ where each $C_i$ is a $(\sigma,\delta)$-polycyclic code. Additionally, recall the definition of the idealizer ring of $Sf$ as $\operatorname{Idl}(Sf) = \{g \in S : fg \in Sf\}$. A polynomial $g \in S$ is called \textit{right invariant} if $gS \subseteq Sg$ and is called \textit{invariant} if $Sg=gS.$ It is easy to see that $x \in \operatorname{Idl}(Sf)$ if and only if $fS\subseteq Sf$. Therefore, Proposition 4.1 in \cite{evaluation} gives the following lemma.

\begin{Lemma}\label{112}  If $x \in \operatorname{Idl}(Sf)$, then the following statements are true.
    \begin{enumerate}
    \item  $f$ is right invariant.
        \item \label{115} $f$ is invariant.
        \item \label{118} If $a\in \mathbb F_q$ and  $f(a)=0$, then $f(\Delta(a))=0.$
        \end{enumerate}
\end{Lemma}
\begin{Lemma}[\cite{LEROY} Corollary 1.12]\label{104}
Let $f\in S$ be a polynomial of degree $n>1$. Then the following statements are equivalent.
\begin{enumerate}
\item $x \in \operatorname{Idl}(Sf)$.
\item For any $g \in S$, $g \in Sf$ if and only if $g(C_f) = 0$.
\item $f(C_f) = 0$.
\end{enumerate}
\end{Lemma}

\begin{Lemma}[\cite{lamleroy1} Proposition 2.6]\label{122} Let $f\in S$ be a polynomial of degree $n.$ Then the following statements hold.
\begin{enumerate}
    \item  The right roots of $f$ can exist in at most $n$ of $(\sigma,\delta)$-conjugacy classes of $\mathbb F_q$. 
    \item If $f(x) = (x-a_0)\cdots(x-a_{n-1})$, then each right root of $f$ in $\mathbb F_q$ is $(\sigma,\delta)$-conjugate to some $a_i$.
\end{enumerate}
\end{Lemma}
 \begin{Lemma}\label{113}
    Let $x \in \operatorname{Idl}(Sf)$ and $f$ be a W-polynomial with a  P-independent set $A=\{a_0,a_1,\ldots, a_{n-1}\}.$ If $\Gamma$ denotes the set of all eigenvalues of $T_{C_f}$, then $\Gamma=\bigcup_{i=0}^{n-1} \Delta(a_i).$
\end{Lemma}
\begin{proof}
 Applying  Theorem \ref{4400}, $f=f_A.$ According to Subsection 
 \ref{123},  $f(x)=(x-b_0) \cdots (x-b_{n-1}),$ where each $b_i$ is $(\sigma, \delta)$-conjugate to some element of $A$. 
 For a fixed  $i$, there is a non-zero polynomial $g_i(x)\in S$ such that $f(x)=g_i(x)(x-b_i).$ Using Lemma \ref{112},  we have $Sf=fS,$ which implies the existence of a non-zero polynomial $h_i(x)\in S$ such that $f(x)=(x-b_i)h_i(x).$  We claim that there is  an element $v\in \mathbb F_q^n$ such that $w_i:=h_i(T_{C_f})(v)\neq 0$. Otherwise, we  have $h_i(T_{C_f}) = 0$. However,  Lemma \ref{104} and the fact $\deg(h_i)<\deg(f)$ lead to $h_i(x) = 0$, which is a contradiction.
  Consequently, there is a non-zero $w_i$ such that $0=f(T_{C_f})(v)=(T_{C_f}-b_iI)h_i(T_{C_f})(v)=(T_{C_f}-b_iI)w_i=T_{C_f}(w_i)-b_iw_iI$,  showing that $b_i$ is an eigenvalue of $T_{C_f}$. Thus, we have $T_{C_f}(\beta v)=\sigma(\beta)T_{C_f}(v)+\delta(\beta)v=(\sigma(\beta)b_i+\delta(\beta))v$   for  $\beta\in \mathbb F_q^*$ and $v\in \mathbb F_q^n.$ Thus  $T_{C_f}(\beta v)=(\sigma(\beta)b_i\beta^{-1}+\delta(\beta)\beta^{-1})\beta v=b_i^\beta \beta v.$ Consequently,  $b_i^\beta$ is an eigenvalue of $T_{C_f}$ for all $\beta \in \mathbb{F}_q^*$, showing $\bigcup_{i=0}^{n-1} \Delta(b_i) \subseteq \Gamma$. Conversely, let $\alpha\in \Gamma$. Then,  there exists a non-zero vector $ v\in \mathbb F_q$ such that $T_{C_f}(v)=\alpha v$. By Proposition 4.3 in \cite{evaluation}, we have $0=f(T_{C_f})(v) = f(\alpha)v$, which implies $f(\alpha)=0.$  Utilizing Lemma \ref{122}, we conclude that  $\alpha\in \bigcup_{i=0}^{n-1} \Delta(b_i).$ Therefore, $ \Gamma=\bigcup_{i=0}^{n-1} \Delta(b_i).$ Since each $b_i$ is $(\sigma, \delta)$-conjugate to some element of $A$, we have $\Gamma=\bigcup_{i=0}^{n-1} \Delta(a_i).$
  
\end{proof}
A  $(\sigma,\delta)$-Pseudo Linear Transformation $T : V \to V$ is called \textit{algebraic} if there is a polynomial  $g(x)=\sum_{i=0}^{n-1}g_i x^i$  such that $g(T)=\sum_{i=0}^{n-1}g_i T^i=0.$ The minimal polynomial possessing this property is called the minimal polynomial of $T$ and is denoted by  $g_T$, see \cite{evaluation}. Let  $\mathbb F_q^n$ be equipped with the  basis $\mathcal D = \{e_1, e_2, \cdots, e_n\}$, where each $e_i$ is a vector with a $1$ in position $i$ and $0$ in all other positions. With the facts that $\delta(1) = 0 = \delta(0)$ and $\sigma(1) = 1$, it is easily to see that  $[T_{C_f}]_{\mathcal D} = C_f$. The main theorem of this section follows, providing a decomposition of $\mathbb F_q^n$ into  $T_{C_f}$-invariant submodules associated with distinct eigenvalues of $T_{C_f}$.
\begin{Theorem}\label{140}
    Let $x \in \operatorname{Idl}(Sf)$ and $f$ be a W-polynomial with a  P-independent set $A=\{a_0,a_1,\ldots, a_{n-1}\}.$  Additionally, let  $T_{C_f}: \mathbb F_q^n\to \mathbb F_q^n$ be a $(\sigma,\delta)$-PLT. Then, we have the following statements.
    \begin{enumerate}
        \item \label{120} $\mathbb F_q^n=\oplus_{i=0}^{n-1} V_{\Gamma_i},$ where $V_{\Gamma_i}$ is the vector space spanned by   eigenvectors of $T_{C_f}$ associated with an eigenvalue in $\Gamma_i:=\Delta(a_i)$.
        \item  Each  $V_{\Gamma_i}$ is a $T_{C_f}$-invariant left $S$-submodule of $\mathbb F_q^n$. 
        \item \label{MM2} There are linear transformations $E_i$  such that  $\sum_{i=0}^{n-1}E_i=\operatorname{Id},$ $E_iE_j=0$ for $i\neq j$, and the range of each  $E_i$ is $ V_{\Gamma_i}.$
           \end{enumerate}
    \end{Theorem}
\begin{proof}\
\begin{enumerate}
    \item Applying Theorem \ref{4400}, $C_f$ is $(\sigma,\delta)$-similar to the diagonal matrix $\operatorname{diag} (a_0, a_1, \cdots, a_{n-1})$. Let us define $\mathscr{D}({T_{C_f}}):=\big\{[T_{C_f}]_{\mathcal B}\,:\,\mathcal B\,\, \text{is a basis of }\,\,\mathbb F_q^n\big\}$. According to Proposition 2.4 in \cite{evaluation}, $\mathscr{D}({T_{C_f}})$ is a $(\sigma,\delta)$-similar class of matrices in $M_n(\mathbb F_q)$. Therefore, since $C_f=[T_{C_f}]_{\mathcal D}\in \mathscr{D}({T_{C_f}})$ we have  $\operatorname{diag} (a_0, a_1, \cdots, a_{n-1})\in \mathscr{D}({T_{C_f}})$. Now, applying Theorem 4.9 in \cite{evaluation}, we conclude that $T_{C_f}$ is diagonalizable, meaning that $\mathbb F_q^n=\oplus_{i=0}^{n-1} V_{\Gamma_i}$, where $V_{\Gamma_i}$ is the set of  eigenvectors of $T_{C_f}$ associated with an eigenvalue in $\Gamma_i:=\Delta(a_i).$
    
    \item \label{130} We know $V_{\Gamma_i}$ is the vector space spanned by   eigenvectors of $T_{C_f}$ associated with an eigenvalue in $\Gamma_i:=\Delta(a_i)$.  Let $x\in \mathbb F_q^*$  and  $T_{C_f}(v)=a_i^x v$, then $T_{C_f}(x^{-1}v)=\sigma(x^{-1})T_{C_f}(v)+\sigma(x^{-1})v=\big(\sigma(x^{-1})a_i^x+\sigma(x^{-1}) \big)v=\big(\sigma(x^{-1})a_i^x+\sigma(x^{-1}) \big)xx^{-1}v=(a_i^x)^{x^{-1}}x^{-1}v=a_ix^{-1}v.$ Therefore, without loss of generality, we can  take $ V_{\Gamma_i}=\operatorname{span}\{v\in \mathbb F_q^n: T_{C_f}(v)=a_iv\}$. To prove $V_{\Gamma_i}$ is a $S$-submodule of $\mathbb F_q^n$,  consider  two cases.  In the case  $a_i=0$, then $V_{\Gamma_i}=\operatorname{ker}(T_{C_f}),$ which is a $S$-submodule of $\mathbb F_q^n$ (notice that by Proposition   \ref{5}, $\mathbb F_q^n$ itself is a $S$-module). In the case $a_i\in \mathbb F_q^*,$ we have $T_{C_f}(a_iv)=\sigma(a_i)T_{C_f}(v)+\sigma(a_i)v=\big(\sigma(a_i)a_i+\delta(a_i)\big)a_i^{-1}a_iv=a_i^{a_i}a_iv$, and hence 
    $a_iv\in V_{\Gamma_i}$. Applying Proposition \ref{5},  $x\cdot v=T_{C_f}(v)=a_iv$ for all $v\in V_{\Gamma_i}$, which  follows  $x\cdot v\in V_{\Gamma_i}$. So, the linearity of $V_{\Gamma_i}$ implies that it is a $S$-submodule of $\mathbb F_q^n$. To complete the proof,  let $v\in V_{\Gamma_i}$. Then $T_{C_f}(v)=a_iv\in V_{\Gamma_i}.$ So $V_{\Gamma_i}$ is $T_{C_f}$-invariant.
    \item  By   Lemma 1.11 in \cite{LEROY} and Lemma \ref{104}, we have  $0=f(C_f)=f([T_{C_f}]_{\mathcal D})=[f(T_{C_f})]_{\mathcal D}$ and  $f$ is the minimal polynomial with this property. Therefore,  $f$ is the minimal polynomial of $T_{C_f}$. Now, since $T$ is diagonalizable (as proved in Part \ref{120}), Theorem 4.9 in \cite{evaluation} implies that $f=\prod_{i=0}^{n-1} f_{\Gamma_i},$ where $f_{\Gamma_i}$ is a minimal  polynomial of $\Gamma_i=\Delta(a_i).$ 
     
    Put $h_i=\prod_{j\neq i}f_{\Gamma_j},$ $i=0,\ldots, n-1.$ We claim that $\sum_{i=0}^{n-1} h_iS=S.$ By contradiction, assume that $\sum_{i=0}^{n-1} h_iS=gS$ or $\sum_{i=0}^{n-1} h_iS=Sg$ for some $g\in S$ with  $\operatorname{deg}(g)\geqslant 1.$ If $\sum_{i=0}^{n-1} h_iS=gS$, since $h_1\in \sum_{i=0}^{n-1} h_iS=gS$, then there is $g_1\in S$ such that $h_1=gg_1.$ We know that $h_1$ is the minimal polynomial of the set $\cup_{j\neq 1}\Gamma_j=\cup_{j\neq 1}\Delta(a_j).$ So $h_1$ vanished on $\cup_{j\neq 1}\Gamma_j.$  Thus,  there exists an element $d\in \cup_{j\neq 1}\Gamma_j$ such that $h_1(d)=0$ and $g_1(d)\neq 0.$ Otherwise, $g_1(d)=0$  for all $d\in \cup_{j\neq 1}\Gamma_j,$ that is $h_1$ can divide  $g_1,$ which contradicts with  $h_1=gg_1$ and degree of $g$. Take $x:=g_1(d)\neq 0.$ Using Lemma \ref{20}, we obtain $0=h_1(d)=gg_1(d)=g(d^x)x$, and hence $g(d^x)=0$.  Now assume that $d\in \Gamma _l,$ where $l\neq 1. $ With the same argument, we can say there is a $g_l\in S$ such that $h_l=gg_l.$ The polynomial $h_l$  is a minimal polynomial of $\cup_{j\neq l}\Gamma_j=\cup_{j\neq l}\Delta(a_j)$. Applying  Lemma 5.2 in \cite{lamleroy3},   $h_l$ is right invariant, and applying  Lemma \ref{112}, it is invariant.
      Thus, there is a polynomial $k_l$ such that $h_l=k_lg$. So,  $g(d^x)=0$ leads to $h_l(d^x)=0.$ Using   Lemma \ref{112}  and considering $h_l(d^x)=0$, we obtain $h_l(\Gamma_l)=0,$ which contradicts with  the definition  of $h_1$. With the same argument, we can deduct $\sum_{i=0}^{n-1} h_iS=Sg$ also leads to a contradiction. So, the claim is proved.

      Using this claim,  there are polynomials $g_0(x),\ldots,g_{n-1}(x)$ such that $\sum_{i=0}^{n-1} h_ig_i(x)=1.$  Let us define $E_i:=h_ig_i(T_{C_f})$. Therefore, $\sum_{i=0}^{n-1}E_i=\operatorname{Id}.$ Moreover,
      applying the same proof used for $h_l$, we can show that
      $h_i$ is the right invariant for all $i$.  So, there are polynomials $k_i(x)\in S$ such that $h_ig_i(x)=k_ih_i(x)$. By Corollary \ref{86}, we have 
$E_iE_j=h_ig_ih_jg_j(T_{C_f})=k_ih_ih_jg_j(T_{C_f})=k_i(T_{C_f})h_ih_j(T_{C_f})g_j(T_{C_f})$ for  $i\neq j$.  Since $h_ih_j$ is divided by $f$ and $f(T_{C_f})=0,$ we get $E_iE_j=0.$
      
   To complete the proof, we need to determine the range of $E_i$.  As proved in  Part \ref{130}, $  V_{\Gamma_i} =\operatorname{span}\{v\in \mathbb F_q^n: T_{C_f}(v)=a_iv\}=\operatorname{span}\{v\in \mathbb F_q^n: (T_{C_f}- a_iI)v=0\}.$
 We have $f=\prod_{i=0}^{n-1} f_{\Gamma_i}$  and $n=\deg f$. So, each $f_{\Gamma_i}$ is a linear  polynomial in the form  $(x-b_i)$. On the other side, 
   $f_{\Gamma_i}$ is the minimal polynomial of $\Gamma_i=\Delta_(a_i).$ 
    So $f_{\Gamma_i}(a_i)=0$, which follows that 
   $b_i=a_i$, that is $f_{\Gamma_i}(x)=x-a_i$. Thus, 
$V_{\Gamma_i}=\operatorname{ker}f_{\Gamma_i}(T_{C_f}).$ Now let  $\alpha$ be in the range of $E_i$. Since $\sum _{i=0}^{n-1} E_i=\operatorname{Id}$, $E_jE_i=0$, and $\alpha$ is in the range of $E_i$,   we have  $E_i(\alpha)=\alpha. $   Thus   $f_{\Gamma_i}(T_{C_f})(\alpha)=f_{\Gamma_i}(T_{C_f})E_i(\alpha)=f_{\Gamma_i}(T_{C_f})h_i(T_{C_f})g_i(T_{C_f})(\alpha)=0, $ because  $f=\prod_{i=0}^{n-1} f_{\Gamma_i}$ divides $f_{\Gamma_i}h_ig_i$ and  $f(T_{C_f})=0.$ Hence $V_{\Gamma_i}$ contains the range of $E_i$. Conversely, suppose  $\alpha\in V_{\Gamma_i}.$ Then $f_{\Gamma_i}(T_{C_f})(\alpha)=0.$ Note that for all  $j \neq i$,  $h_jg_j$ is divided by $f_{\Gamma_i}$. So  $E_j(\alpha)  = 0.$  Since $\sum _{i=0}^{n-1} E_i=\operatorname{Id}, $ we can immediately conclude that $E_i(\alpha)=\alpha$, which means  $\alpha$ is in the range of $E_i$.
\end{enumerate}
\end{proof}
\begin{Corollary}\label{MM1}
 Let $x \in \operatorname{Idl}(Sf)$ and $f$ be a W-polynomial with a  P-independent set $A=\{a_0,a_1,\ldots, a_{n-1}\}.$ 
 If $C$ is a $T_{C_f}$-invariant $S$-submodule of $\mathbb F_q^n$,  then 
 \begin{equation}\label{MM3}
   C=\big(C\cap V_{\Gamma_0}\big)\oplus \big(C\cap V_{\Gamma_1}\big)\oplus\cdots\oplus \big(C\cap V_{\Gamma_{n-1}}\big) 
 \end{equation}
and each $S$-submodule $C\cap V_{\Gamma_i}$ is $T_{C_f}$-invariant.

\end{Corollary}
\begin{proof}\
According to  Theorem \ref{140} and its proof,  $\mathbb F_q^n=\oplus_{i=0}^{n-1} V_{\Gamma_i},$ and each projection $E_i$ is associated with a polynomial $p_i$ such that $E_i = p_i(T_{C_f})$.
Let $\alpha \in C$, then $\alpha = \alpha_0 + \alpha_1 + \cdots + \alpha_{n-1}$, where each $\alpha_i \in V_{\Gamma_i}$. Applying Statement \eqref{MM2} in Theorem \ref{140}, we deduce that $\alpha_i = E_i(\alpha) = p_i(T_{C_f})(\alpha)$. Since $C$ is invariant under $T_{C_f}$,  it follows that  $\alpha_i\in C.$ In other words,  $\alpha_i\in C\cap V_{\Gamma_i},$ and hence $C\subseteq \big(C\cap V_{\Gamma_0}\big)\oplus \big(C\cap V_{\Gamma_1}\big)\oplus\cdots\oplus \big(C\cap V_{\Gamma_{n-1}}\big),$ which completes the proof of Equation \eqref{MM3}. Applying  Theorem \ref{140},   each  $V_{\Gamma_i}$ is a $T_{C_f}$-invariant left $S$-submodule of $\mathbb F_q^n$. Thus each $C\cap V_{\Gamma_i}$ is $T_{C_f}$-invariant.
 \end{proof}
 \begin{Corollary}
     If $x \in \operatorname{Idl}(Sf)$ and $C$ is a simple-root $(\sigma,\delta)$-polycyclic code, then  there are simple-root $(\sigma,\delta)$-polycyclic codes $C_i$  such that $$C=C_0\oplus C_1\oplus \cdots\oplus C_{n-1}.$$
 \end{Corollary}
 \begin{proof}
     Take $C_i=C\cap V_{\Gamma_i} $ and use Corollary \ref{MM1}. 
 \end{proof}

\section{Conclusions and future work}

In this paper, we have explored $(\sigma,\delta)$-polycyclic codes in the Ore extension $S=R[x,\sigma,\delta]$, where $R$ is a finite, not necessarily commutative, ring. We have defined these codes as S-submodules in the quotient module $S/Sf$. 
We have established that the Euclidean duals of $(\sigma,\delta)$-polycyclic codes are $(\sigma,\delta)$-sequential codes.
Using $(\sigma,\delta)$-Pseudo Linear Transformation, we have introduced the annihilator dual of $(\sigma,\delta)$-polycyclic codes.  Moreover, we have demonstrated that the annihilator dual of $(\sigma, \delta)$-polycyclic codes preserves their $(\sigma, \delta)$-polycyclic nature. Furthermore,  we have classified the Hamming isometrical equivalence  $(\sigma,\delta)$-polycyclic codes, a crucial aspect in the study of each class of codes. Employing Wedderburn polynomials, we have defined simple-root $(\sigma,\delta)$-polycyclic codes and have constructed the $(\sigma,\delta)$-Mattson-Solomon transform for this class of codes. Finally,  we have provided a decomposition of the simple-root $(\sigma,\delta)$-polycyclic codes associated with distinct eigenvalues of $T_{C_f}$. To further our research we plan to explore $(\sigma,\delta)$-multivariable codes and derive optimal codes from $(\sigma,\delta)$-polycyclic codes.

\end{document}